\documentclass[a4paper,11pt]{amsart}
\usepackage{amsmath,amssymb}
\usepackage{mathrsfs}
\usepackage{eucal}
\usepackage{cases}
\usepackage{cite}
\usepackage{comment,color}

\makeatletter
 
  \@addtoreset{equation}{section}
 \makeatother

\newcommand{\slim}{\operatorname{s-}\hspace{-0.25cm}\lim_{t\to\pm\infty}}

\newcommand{\Schr}{Schr\"odinger }

\renewcommand{\a}{\alpha}
\renewcommand{\b}{\beta}
\newcommand{\g}{\gamma}
\renewcommand{\d}{\delta}
\newcommand{\e}{\varepsilon}
\newcommand{\f}{\varphi}
\newcommand{\ps}{\psi}
\newcommand{\z}{\zeta}

\renewcommand{\l}{\lambda}

\newcommand{\x}{\xi}
\newcommand{\y}{\eta}
\newcommand{\s}{\sigma}

\newcommand{\G}{\Gamma}

\newcommand{\C}{\mathbb{C}}
\newcommand{\R}{\mathbb{R}}
\newcommand{\Z}{\mathbb{Z}}
\newcommand{\T}{\mathbb{T}}

\newcommand{\HS}{\mathcal{H}}

\newcommand{\la}{\langle}
\newcommand{\ra}{\rangle}
\newcommand{\jap}[1]{\langle #1 \rangle}

\renewcommand{\ker}{\operatorname{Ker}}
\newcommand{\supp}{\operatorname{supp}}

\newtheorem{thm}{Theorem}[section]
\newtheorem{lem}[thm]{Lemma}
\newtheorem{prop}[thm]{Proposition}
\newtheorem{cor}[thm]{Corollary}

\theoremstyle{definition}
\newtheorem{defn}[thm]{Definition}
\newtheorem{ass}[thm]{Assumption}

\theoremstyle{remark}
\newtheorem{rem}[thm]{Remark}

\title[Long-range scattering for general lattices]{Construction of Isozaki-Kitada modifiers for discrete Schr\"odinger operators on general lattices}
\author{Yukihide TADANO}
\thanks{\noindent Keywords:long-range scattering theory, discrete Schr\"odinger operators, modified wave operators, time-independent modifiers\\
Mathematics subject classification:47A40, 47B39, 81U05\\
Department of Mathematics, 
Tokyo University of Science, 
Kagurazaka 1-3, Shinjuku-ku, Tokyo 162-8601, Japan. \\
E-mail: {\tt y.tadano@rs.tus.ac.jp}}
\date{\today}

\begin{document}
\begin{abstract}
We consider a scattering theory for difference operators on $\mathcal{H}=\ell^2(\mathbb{Z}^d; \mathbb{C}^n)$ perturbed with a long-range potential $V:\Z^d\to\R^n$.
One of the motivating examples is discrete Schr\"odinger operators on $\Z^d$-periodic graphs.
We construct time-independent modifiers, so-called Isozaki-Kitada modifiers, and we prove that the modified wave operators with the above-mentioned Isozaki-Kitada modifiers exist and that they are complete.
\end{abstract}

\maketitle


\section{Introduction}

The aim of the present article is to construct a long-range scattering theory for difference operators on the space of vector-valued functions on $\Z^d$.
This problem is motivated by discrete Schr\"odinger operators on an arbitrary non-primitive lattice, e.g., hexagonal lattice, diamond lattice, Kagome lattice and graphite (see \cite{An-I-M} for more examples).
Note that the cases of primitive lattices and the hexagonal lattice are considered in \cite{T1} and \cite{T2}, respectively.

Let $\mathcal{H}=\ell^2(\mathbb{Z}^d; \mathbb{C}^n)$, where $d$ and $n$ are positive integers. For $u \in \mathcal{H}$, we use the notation
\begin{align*}
u=\left( \begin{array}{c}
u_1 \\ u_2 \\ \vdots \\ u_n
\end{array}
\right), \quad u_j \in \ell^2(\mathbb{Z}^d)=\ell^2(\mathbb{Z}^d;\C).
\end{align*}
We consider a generalized form of discrete \Schr operators on $\mathcal{H}$:
\begin{align*}
H=H_0+V .
\end{align*}
The unperturbed operator $H_0$ is defined as a convolution operator by $(f_{jk})_{1\leq j,k\leq n}$, that is,
\begin{align*}
&H_0u = \left( \begin{array}{cccc}
H_{0,11} & H_{0,12} & \cdots & H_{0,1n} \\
H_{0,21} & H_{0,22} & \cdots & H_{0,2n} \\
\vdots & \vdots & \ddots & \vdots \\
H_{0,n1} & H_{0,n2} & \cdots & H_{0,nn}
\end{array}
\right) u, \quad u \in \mathcal{H}, \\
&H_{0,jk}u_k(x)=\sum_{y\in\Z^d}f_{jk}(x-y)u_k(y), \quad u_k \in \ell^2(\Z^d).
\end{align*}
Here each $f_{jk}: \mathbb{Z}^d \rightarrow \mathbb{C}$ is a rapidly decreasing function, i.e.,
\begin{align*}
\sup_{x\in\Z^d} \la x \ra^m |f_{jk}(x)| <\infty
\end{align*}
for any $m\in\mathbb{N}$, where $\la x \ra = (1+|x|^2)^{\frac{1}{2}}$.
The perturbation $V$ is a multiplication operator by $V={}^t(V_1,\cdots,V_n):\Z^d \to \R^n$,
\begin{align*}
Vu(x) &= \left( \begin{array}{c}
V_1(x) u_1(x) \\ V_2(x) u_2(x) \\ \vdots \\ V_n(x) u_n(x)
\end{array}
\right) , \quad u \in \mathcal{H} . 
\end{align*}

We denote the discrete Fourier transform by $\mathcal{F}$; 
\begin{align*}
\mathcal{F}u(\xi)=&\left(\begin{array}{c}
Fu_1(\xi) \\
Fu_2(\xi) \\
\vdots \\
Fu_n(\xi)
\end{array}
\right), \quad \xi \in \mathbb{T}^d:=[-\pi,\pi)^d , \\
F u_j(\x) =& (2\pi)^{-\frac{d}{2}}\sum_{x\in\Z^d}e^{-ix\cdot\x}u_j(x) ,
\end{align*}
for $u\in \ell^1(\Z^d;\C^n)$.
Then $\mathcal{F}$ is extended to a unitary operator from $\mathcal{H}$ onto $\hat{\mathcal{H}}=L^2(\mathbb{T}^d; \mathbb{C}^n)$, and we denote its extension by the same symbol $\mathcal{F}$.
We easily see that $\mathcal{F} \circ H_0 \circ \mathcal{F}^*$ is the multiplication operator on $\T^d$ by the matrix-valued function
\begin{align*}
H_0(\xi) = \left( \begin{array}{cccc}
h_{11}(\xi) & h_{12}(\xi) & \cdots & h_{1n}(\xi) \\
h_{21}(\xi) & h_{22}(\xi) & \cdots & h_{2n}(\xi) \\
\vdots & \vdots & \ddots & \vdots \\
h_{n1}(\xi) & h_{n2}(\xi) & \cdots & h_{nn}(\xi)
\end{array}
\right),
\end{align*}
where
\begin{align*}
h_{jk}(\xi):=\sum_{x\in\Z^d}e^{-ix\cdot\xi}f_{jk}(x).
\end{align*}
Since $f_{jk}$'s are assumed to be rapidly decreasing, $h_{jk}$'s are smooth functions on $\T^d$.
Note that $\sigma(H_0) = \{ \lambda \mid \det (H_0(\xi)-\lambda) =0 \ \hbox{for some}\  \xi \in \mathbb{T}^d \} $ and $H_0$ is a self-adjoint operator if and only if $H_0(\x)$ is a symmetric matrix for any $\x\in\T^d$, i.e., by the definition of $H_0(\x)$, 
\begin{align}\label{selfadj}
\overline{f_{jk}(-x)}=f_{kj}(x) , \quad x \in \mathbb{Z}^d, \ 1 \leq j,k \leq n. 
\end{align}

In this paper, we assume the following assumption concerning the self-adjointness of $H_0$ and a long-range condition of $V$.

\begin{ass}\label{ass}
(1) $f_{jk}$'s are rapidly decreasing functions satisfying \eqref{selfadj}. 

(2) $V={}^t(V_1, \cdots, V_n)$ has the following representation
\begin{align*}
V=V_L + V_S ,
\end{align*}
where each entry of $V_L$ is the same, i.e., $V_L={}^t(V_\ell, \cdots, V_\ell)$ with some $V_\ell : \Z^d \to \R$.
Furthermore, there exist $\rho>0$ and $C,C_\a>0$ such that
\begin{align}
|\tilde\partial_x^\a V_\ell(x)| \leq C_\a \la x \ra^{-\rho-|\a|} , \label{long range} \\
|V_S(x)| \leq C \la x \ra^{-1-\rho} \label{short range}
\end{align}
for any $x\in\Z^d$ and $\a\in\Z_+^d$. Here $\tilde\partial_x^\a=\tilde\partial_{x_1}^{\a_1} \cdots \tilde\partial_{x_d}^{\a_d}$, $\tilde \partial_{x_j} V(x)=V(x)-V(x-e_j)$ is the difference operator with respect to the $j$-th variable.
\end{ass}


Assumption \ref{ass} implies that $V$ is a compact operator on $\mathcal{H}$ and hence
\begin{align*}
\sigma_{\operatorname{ess}}(H)=\sigma_{\operatorname{ess}}(H_0) ,
\end{align*}
where $\sigma_{\operatorname{ess}}(H)$ (resp.\ $\sigma_{\operatorname{ess}}(H_0)$) denotes the essential spectrum of $H$ (resp.\ $H_0$).

We denote the union of Fermi surfaces corresponding to the energies in $\G\subset\R$ by
\begin{align*}
\operatorname{Ferm}(\G):=&\{ p = (\x,\l)\in\T^d\times\G \mid \l \text{ is an eigenvalue of }H_0(\x) \} \\
=& \{ p = (\x,\l)\in\T^d\times\G \mid \det \left(H_0(\x)-\l\right)=0 \} . \nonumber
\end{align*}
Before describing the main theorem, we prepare the notation of non-threshold energies.

\begin{defn}\label{nonthr}
$\l_0\in\s(H_0)$ is said to be \textit{a non-threshold energy of $H_0$} if the following properties (1) and (2) hold:

(1) For any $\x_0\in\T^d$ such that $\det(H_0(\x_0)-\l_0)=0$, there exists an open neighborhood $G \subset \T^d \times \R$ of $p=(\x_0,\l_0)$ such that $\operatorname{Ferm}(\R) \cap G$ has a graph representation, i.e. 
\begin{align}\label{lambda}
\operatorname{Ferm}(\R) \cap G = \{(\x,\l(\x)) \mid \x\in U \}
\end{align}
with some $U\ni \x_0$ and $\l\in C^\infty(U)$.


(2) Let $\x_0$ be arbitrarily fixed so that $\det(H_0(\x_0)-\l_0)=0$ holds, and let $\l(\x)$ be as in \eqref{lambda}.
Then $\nabla_\x \l(\x_0) \neq 0$ holds.
\end{defn}

\begin{rem}
There is a sufficient condition of non-threshold energies:
\begin{align*}
\nabla_\x \det(H_0(\x)-\l_0)\neq0 \ \text{for any } \x\in\T^d \ \text{such that} \ \det(H_0(\x)-\l_0)=0 .
\end{align*}
The principal difference is that Definition \ref{nonthr} covers the case where $H_0(\x)$ has degenerate eigenvalues but no branching occurs.
\end{rem}

Let $\Gamma(H_0)$ be the set of non-threshold energies of $H_0$.
Then $\Gamma(H_0)$ is an open set of $\R$ and $\Gamma(H_0) \subset \s(H_0)$.
Note that $H_0$ has purely absolutely continuous spectrum on $\G(H_0)$, i.e., $\s_{pp}(H_0)\cap \G(H_0)=\s_{sc}(H_0)\cap \G(H_0)=\phi$ (see Remark \ref{rem purely ac}).

The main theorem of this paper is the following.

\begin{thm}\label{main}
Suppose Assumption \ref{ass} and $\Gamma \Subset \Gamma(H_0)$. Then there are bounded operators $J_\pm=J_{\pm,\Gamma}$ on $\mathcal{H}$, called Isozaki-Kitada modifiers, such that the modified wave operators exist:
\begin{align}\label{modified WOs}
W_{\textrm{IK}}^\pm(\G)=\slim e^{itH}J_\pm e^{-itH_0}E_{H_0}(\Gamma),
\end{align}
where $E_{H_0}$ denotes the spectral measure of $H_0$, and that the following properties hold:

i) Intertwining property: $H W_{\textrm{IK}}^\pm(\G)=W_{\textrm{IK}}^\pm(\G) H_0$.

ii) Partial isometries: $\|W_{\textrm{IK}}^\pm(\G) u\|=\|E_{H_0}(\Gamma)u\|$.

iii) Completeness: $\operatorname{Ran} W_{\textrm{IK}}^\pm(\G) = E_{H}(\Gamma)\mathcal{H}_{\text{ac}}(H)$.

\noindent
Here $\mathcal{H}_{\text{ac}}(H)$ denotes the absolutely continuous subspace of $H$.
\end{thm}

Various examples of unperturbed operators $H_0$ are given by Ando, Isozaki and Morioka \cite[Section 3]{An-I-M}.
Note that, if the perturbation $V$ is short-range, i.e.,\ $V_L=0$, we can set $J_\pm=\operatorname{Id}_\mathcal{H}$, thus there exist the wave operators in this case.
See \cite{P-R} for short-range scattering theory for discrete \Schr operators on various lattices.
We also note that a long-range scattering theory in the case of $n=1$, e.g.,\ discrete \Schr operators on square and triangular lattices, is considered by Nakamura \cite{N} and the author \cite{T1}.
Moreover, Theorem \ref{main} covers an arbitrary periodic lattice $\mathcal{L}$ with each primitive unit cell $\mathcal{L}/\Gamma$ containing finite elements, where $\Gamma \cong \Z^d$ denotes the transformation group associated to $\mathcal{L}$.
In particular, it includes the result by the author \cite{T2}, where a long-range scattering theory for discrete \Schr operators on the hexagonal lattice is studied.
See also \cite{D-G}, \cite{R-S3}, \cite{Y} and references therein for scattering theory of \Schr operators on $\R^d$.










The organization of this paper is as follows. 
We first prepare notations and properties of pseudodifference operators in Section \ref{notations sec}.
In Section \ref{Mourre section}, the  limiting absorption principle and the propagation estimate for $H$ are studied.
We use the Mourre theory and a standard argument of the propagation of wave packets as in Yafaev \cite[Chapter 10]{Y}.
The construction of conjugate operators is essentially due to Parra and Richard \cite{P-R}.
Section \ref{classical section} is devoted to constructing phase functions which are given as local solutions to eikonal equations corresponding to each fiber of eigenvalues of $H_0(\x)$.
The construction of phase functions is due to \cite{N3}.
In Section \ref{IK section}, using the phase functions in the previous section, we construct Isozaki-Kitada modifiers.
Finally in Section \ref{proof section}, we use lemmas in the previous section to prove Theorem \ref{main}.
The proof is based on Kato's smooth perturbation theory, and is an analogue of that in long-range scattering theory for \Schr operators on $\R^d$ (see \cite{Y}).

\section*{acknowledgement}
The author was partially supported by JSPS Grant Numbers 20J00247, 21K20337 and 23K12991.
This work is supported by the Research Institute for Mathematical Sciences, an International Joint Usage/Research Center located in Kyoto University.
The author would like to thank Shu Nakamura for encouraging me to write this article.


\section{Preliminaries}\label{notations sec}

\subsection{Representations of fibers}

Let $\G$ be as in Theorem \ref{main}, and let $I \Subset \G(H_0)$ be fixed so that $\G \Subset I \Subset \G(H_0)$.

For each $p=(\x_0,\l_0) \in \operatorname{Ferm}(\G(H_0))$, let $G=G_p$ be as in Definition \ref{nonthr}.
Then $\{ G_p \}_{p\in\operatorname{Ferm}(\G(H_0))}$ is an open covering of $\operatorname{Ferm}(\G(H_0))$.
Since $\operatorname{Ferm}(\overline{I})$ is compact, we can take a finite family $\{ G_{j} \}_{j=1}^J = \{ G_{p_j} \}_{j=1}^J$ of open sets which covers $\operatorname{Ferm}(\overline{I})$.

Note that $\{ G_j \cap \operatorname{Ferm}(\R) \}_{j=1}^J$ is also a covering family of $\operatorname{Ferm}(\overline{I})$.
Let $G'_k$, $k=1,\dots,K$, be the connected components of $\cup_{j=1}^J G_j \cap \operatorname{Ferm}(\R)$.
We see that each $G'_k$ remains to have a graph representation
\begin{align}\label{G_k representation}
G'_k=\{ (\x,\l_k(\x)) \mid \x \in \mathcal{U}_k \}
\end{align}
with some open set $\mathcal{U}_k \subset \T^d$ and $\l_k\in C^\infty(\mathcal{U}_k)$.
We denote by $P_k(\x)$ the projection matrix onto $\ker (H_0(\x)-\l_k(\x))$ for $\x\in \mathcal{U}_k$.
Then we have for $\psi\in C_c^\infty(I)$
\begin{align}\label{Rep of E_H0}
\ps(H_0(\x)) = \sum_{k=1}^K \ps(\l_k(\x)) P_k(\x) \chi_{\mathcal{U}_k}(\x) .
\end{align}

%

\subsection{Pseudodifference calculus}

For $a:\Z^d \times \T^d \to M_n(\C)\cong \C^{n\times n}$,
\begin{align*}
a(x,D_x)u(x) :=(2\pi)^{-\frac{d}{2}}\int_{\T^d} e^{ix\cdot\x} a(x,\x) \mathcal{F}u(\x)d\x , \quad u\in \mathcal{H} ,
\end{align*}
denotes the pseudodifference operator on $\Z^d$ with symbol $a(x,\x)$.
If $a$ depends only on $\x$, we denote by $a(D_x)=\mathcal{F}^*\circ a(\cdot) \circ \mathcal{F}$ the Fourier multiplier associated with $a(\x)$ in short.

We cite a lemma concerning the pseudodifference calculus on $\mathcal{H}$ (see \cite[Theorem 4.2.10]{R-T} and the proof of \cite[Lemma 2.2]{T1}).

\begin{lem}\label{PDO lemma}
Let $a:\Z^d\times\Z^d\times\T^d \to M_n(\C)$ be a smooth function with respect to $\T^d$, and let
\begin{align*}
Au(x) = (2\pi)^{-d} \int_{\T^d} \sum_{y\in\Z^d} e^{i(x-y)\cdot\x} a(x,y,\x) u(y)d\x .
\end{align*}
Suppose that for any $\a\in\Z_+^d$
\begin{align}\label{S^0}
\sup_{(x,y,\x)\in\Z^d\times\Z^d\times\T^d}|\partial_{\x}^\a a(x,y,\x)| < \infty .
\end{align}
Then $A$ is a bounded operator on $\ell^2(\mathcal{H})$.
\end{lem}

Let $S^m$ be the symbol class of order $m\in\R$, i.e.,
\begin{align*}
S^m=\left\{ a: \Z^d\times\T^d \to M_n(\C) \mid a(x,\cdot)\in C^\infty(\T^d;M_n(\C)),\ \forall x\in\Z^d , \right. \\
 \left. \sup_{(x,\x)\in\Z^d\times\T^d} \jap{x}^{-m+|\a|} |\tilde\partial_{x}^\a\partial_{\x}^\b a(x,\x)| < \infty ,\ \forall \a,\b\in\Z_+^d \right\} ,
\end{align*}
where $\tilde\partial_{x}^\a$ denotes the difference operator as in \eqref{long range}.

The following two assertions are analogous to the composition formula for pseudodifferential operators.
See \cite[Theorems 4.7.3 and 4.7.10]{R-T} for the proofs.

\begin{lem}\label{composition lemma}
Let $a\in S^m$ and $b\in S^{\ell}$.
Then $a(x,D_x)b(x,D_x)=c(x,D_x)$ with some $c\in S^{m+\ell}$ satisfying the asymptotic expansion
\begin{align*}
&c(x,\x) - \sum_{|\a| \leq M} \partial_{\x}^\a a(x,\x) \tilde\partial_{x}^\a b(x,\x) \in S^{m+\ell-M-1}
\end{align*}
for any $M\in\Z_+$.
\end{lem}

\begin{lem}\label{adjoint lemma}
Let $a\in S^m$. 
Then there exists $b\in S^m$ such that $a(x,D_x)^*=b(x,D_x)$ and $b(x,\x) - a(x,\x)^* \in S^{m-1}$.
\end{lem}

\subsection{Kato's smooth perturbation theory}
For a self-adjoint operator $H$ and an $H$-bounded operator $G$, we say that $G$ is $H$-smooth if
\begin{align}\label{Kato smooth}
\frac{1}{2\pi}\sup_{\|u\|_{\mathcal{H}}=1,u\in D(H)} \int_{-\infty}^\infty \left\| Ge^{-itH}u \right\|^2 dt < \infty.
\end{align}
For a Borel set $I\subset\R$, we say that $G$ is $H$-smooth on $I$ if $GE_H(I)$ is $H$-smooth, and we also say that $G$ is locally $H$-smooth on $I$ if $G$ is $H$-smooth on $I'$ for any $I' \Subset I$.

There are several conditions equivalent to \eqref{Kato smooth} (see e.g.\ \cite{Yg}), and the one we need in the following is:
\begin{align}\label{Kato smooth 2}
\sup_{\lambda\in\R,\e>0} \|G\d_\e(\lambda,H)G^*\| <\infty,
\end{align}
where $\d_\e(\lambda,H)=\frac{1}{2\pi i}\{(H-\lambda-i\e)^{-1}-(H-\lambda+i\e)^{-1}\}$.

\section{Limiting absorption principle and radiation estimates}\label{Mourre section}

In this section, we consider the limiting absorption principle and radiation estimates for the proof of Theorem \ref{main}.

\subsection{Limiting absorption principle}

For a self-adjoint operator $A$ and $m\in\mathbb{N}$, let
\[
C^m(A)=\{ S\in\mathcal{B}(\mathcal{H}) \mid \R\to\mathcal{B}(\mathcal{H}),t\mapsto e^{-itA}Se^{itA} \ \text{is strongly of class} \ C^m \},
\]
and $C^\infty(A)=\cap_{m\in\mathbb{N}}C^m(A)$.
We denote by $\mathcal{C}^{1,1}(A)$ the set of the operators $S$ satisfying
\[
\int_0^1 \| e^{-itA} S e^{itA}+e^{itA} S e^{-itA}-2S \|\frac{dt}{t^2} < \infty .
\]

We set the Besov space
\begin{align*}
B:=(\mathcal{D}(\jap{x}),\mathcal{H})_{\frac{1}{2},1},
\end{align*}
where we have used the notation of real interpolation $(\cdot,\cdot)_{\theta,p}$ between Banach spaces (see \cite[Section 2.1]{A-BdM-G}).

The following proposition is called the limiting absorption principle.
The proof is given by the Mourre theory, and the construction of conjugate operators is essentially due to \cite[Lemma 6.2]{P-R}.

\begin{prop}\label{LAP prop}
Suppose Assumption \ref{ass}. Then:

(1) The set of eigenvalues of $H$ is locally finite in $\G(H_0)$ with counting multiplicities.

(2) For any $\lambda \in \G(H_0)\backslash \s_{pp}(H)$, there exist the weak-* limits in $\mathcal{B}(B,B^*)$
\begin{align*}
\operatorname{w^*-}\lim_{\e\to+0} (H-\l\mp i\e)^{-1} .
\end{align*}
Moreover, each convergence is locally uniform in $\l\in\G(H_0)\backslash \s_{pp}(H)$.
In particular, for any $\G\Subset\G(H_0)\backslash \s_{pp}(H)$,
\begin{align}\label{LAP}
\sup_{\l \in \G, \e>0} \|(H-\l\mp i\e)^{-1}\|_{\mathcal{B}(B,B^*)} < \infty .
\end{align}
\end{prop}

\begin{proof}
Let $\G\Subset\G(H_0)$ be arbitrarily fixed, and recall the representation \eqref{Rep of E_H0}.
We set $\chi_k \in C_c^\infty(\mathcal{U}_k)$ so that $\chi_k=1$ on $\lambda_k^{-1}(\G)$.
We also set the conjugate operator $A$ by
\begin{align*}
A=& \sum_{k=1}^K P_k(D_x) \chi_k(D_x) i[\l_k(D_x),|x|^2] P_k(D_x) \chi_k(D_x) \\
=& \sum_{k=1}^K P_k(D_x) \chi_k(D_x) M_k P_k(D_x) \chi_k(D_x) , \nonumber
\end{align*}
where
\begin{align*}
M_k = x \cdot \nabla_\x\l_k(D_x) + \nabla_\x\l_k(D_x) \cdot x .
\end{align*}
Now we employ the Mourre theory (\cite[Proposition 7.1.3, Corollary 7.2.11, Theorem 7.3.1]{A-BdM-G}, see also \cite[Theorem A.1]{T2}).
Then, since $A$ is $\jap{x}$-bounded, it suffices to show that $H\in \mathcal{C}^{1,1}(A)$ and that, for any $\psi \in C_c^\infty(\G)$, 
there exist $c>0$ and a compact operator $K$ such that the Mourre inequality holds:
\begin{align}
\psi(H) i[H,A] \psi(H) \geq c \psi(H)^2 + K . \label{Mourre}
\end{align}

For the first assertion, we easily see $H_0\in C^\infty(A)$, and $V \in \mathcal{C}^{1,1}(A)$ is proved by \eqref{long range}, \eqref{short range} and Lemma \ref{composition lemma} (see \cite{P-R} and \cite{T2} for details of the proof).

For the proof of \eqref{Mourre}, we learn by Definition \ref{nonthr} (2) that
\begin{align} \label{free Mourre}
&\psi(H_0) i[H_0,A] \psi(H_0) \\
=& 2 \sum_{k=1}^K P_k(D_x) \ps(\l_k(D_x) ) \chi_k(D_x) |\nabla_\x \lambda_k(D_x)|^2 P_k(D_x) \ps(\l_k(D_x) ) \chi_k(D_x) \nonumber\\
\geq& c \sum_{k=1}^K P_k(D_x) \ps(\l_k(D_x) )^2 \chi_k(D_x)^2 \geq c \ps(H_0)^2 . \nonumber
\end{align}
It follows from \eqref{long range} and \eqref{short range} that $i[V,A]$ and $\psi(H)-\psi(H_0)$ are compact, and hence we have \eqref{Mourre}.

\end{proof}

\begin{rem}\label{rem purely ac}
If we adopt the Mourre theory to $H=H_0$, \eqref{free Mourre} implies that $H_0$ has purely absolutely continuous spectrum on $\G(H_0)$.
\end{rem}

Since $B \supset \jap{x}^s\mathcal{H}$ and $B^* \subset \jap{x}^{-s}\mathcal{H}$ hold for any $s>\frac{1}{2}$ (see, e.g., \cite[Theorem 3.4.1]{A-BdM-G}), 
\eqref{LAP} and the equivalence between \eqref{Kato smooth} and \eqref{Kato smooth 2} imply the following corollary.

\begin{cor}
For any $s>\frac{1}{2}$, $\jap{x}^{-s}$ is locally $H$-smooth on $\G(H_0) \backslash \s_\textrm{pp}(H)$.
\end{cor}

\subsection{Radiation estimates}
In order to prove the existence and completeness of modified wave operators, we use, in addition to the limiting absorption principle, other propagation estimates called radiation estimates (see \cite[Theorem 10.1.7]{Y}).

\begin{prop}\label{prop_rad_est}
Let $\G\Subset\G(H_0)$ be fixed, and let $\l_k(\x)$, $k=1,\dots,K$, be as in \eqref{G_k representation}.
We set for $k=1,\dots,K$ and $j=1,\dots,d$,
\begin{align*}
\nabla_{k,j}^\perp := \left\{ (\partial_{\x_j}\l_k)(D_x) - \chi_{\{x\neq0\}} |x|^{-2} x_j \la x, (\nabla_\x \l_k)(D_x)\ra \right\} P_k(D_x)\chi_k(D_x) ,
\end{align*}
where $\chi_k \in C_c^\infty(\mathcal{U}_k)$ is fixed arbitrarily so that $\chi_k=1$ on $\lambda_k^{-1}(\G)$.
Then
\begin{align}\label{radiation op}
\chi_{\{x\neq0\}} |x|^{-\frac{1}{2}} \nabla_{k,j}^\perp
\end{align}
is locally $H$-smooth on $\G(H_0)\backslash \s_{pp}(H)$.
\end{prop}

\begin{proof}
Fix $k=1,\dots,K$.
For simplicity of notation, we write $\l$, $P$, $\chi$ and $\nabla_j^\perp$ instead of $\l_k$, $P_k$, $\chi_k$ and $\nabla_{k,j}^\perp$, respectively.

Let $a\in C^\infty(\R^d)$ be fixed so that $a(x)=|x|$ for $|x|\geq1$,
and let
\begin{align*}
a_j:=\partial_{x_j}a, \quad v_j:=\partial_{\x_j}\l.
\end{align*}
We set
\begin{align*}
\mathbb{A} := (P\chi)(D_x) \sum_{j=1}^d \left\{a_j(x)v_j(D_x)+v_j(D_x)a_j(x) \right\} (P\chi)(D_x) .
\end{align*}
Then the representation \eqref{Rep of E_H0} implies
\begin{align*}
i[H_0,\mathbb{A}] = (P\chi)(D_x) \cdot M \cdot (P\chi)(D_x),
\end{align*}
where
\begin{align*}
M = \sum_{j=1}^d \left\{ i[\l(D_x),a_j(x)] \cdot v_j(D_x) + v_j(D_x) \cdot i[\l(D_x),a_j(x)] \right\} .
\end{align*}
It follows from Lemma \ref{composition lemma} that, formally,
\begin{align*}
M 
&= 2 \sum_{j=1}^d \sum_{\ell=1}^d v_\ell(D_x)a_{j\ell}(x) v_j(D_x) + R_1 ,
\end{align*}
where $a_{j\ell}:=\partial_{x_\ell}\partial_{x_j}a$, and $R_1$ satisfies $\jap{x}^2 (P\chi)(D_x) R_1 (P\chi)(D_x) \in \mathcal{B}(\mathcal{H})$.
Since for $|x|\geq1$
\begin{align*}
a_{j\ell}(x) = \partial_{x_\ell}\partial_{x_j}\left(|x|\right)
 = -\frac{x_j x_\ell}{|x|^3} + \d_{j\ell} |x|^{-1} , 
\end{align*}
we learn
\begin{align}\label{radiation pf commutator}
&(u,i[H_0,\mathbb{A}]u) \\
 =& -2 \sum_{j=1}^d \sum_{\ell=1}^d (u^\ell, \frac{x_j x_\ell}{|x|^3} \chi_{\{x\neq0\}} u^j) 
+2 \sum_{j=1}^d (u^j, |x|^{-1}\chi_{\{x\neq0\}}u^j) \nonumber \\
& \hspace{55mm} + ((P\chi)(D_x)u, R_2 (P\chi)(D_x)u), \nonumber
\end{align}
where
\begin{align*}
u^j:= (v_j P\chi)(D_x) u,
\end{align*}
and
\begin{align*}
R_2= R_1 + 2 \sum_{j=1}^d \sum_{\ell=1}^d a_{j\ell}(0) v_\ell(D_x) \chi_{x=0}(x) v_j(D_x)
\end{align*}
also satisfies $\jap{x}^2 (P\chi)(D_x) R_2 (P\chi)(D_x) \in \mathcal{B}(\mathcal{H})$.

On the other hand, 
a direct computation implies for $x\neq0$
\begin{align*}
&\left| \nabla_{j}^\perp u (x) \right|^2 \\
=& |u^j(x)|^2 - |x|^{-2} x_j \sum_{\ell=1}^d x_\ell \left( u^\ell(x)\overline{u^j(x)}+\overline{u^\ell(x)}u^j(x) \right) \\
& \hspace{5cm} + |x|^{-4} {x_j}^2 \sum_{\ell=1}^d \sum_{m=1}^d x_\ell x_m \overline{u^\ell(x)}u^m(x).
\end{align*}
Summing up over $j=1,\dots,d$, we learn
\begin{align}\label{differentiation across trajectory}
&\sum_{j=1}^d \left| \nabla_{j}^\perp u (x) \right|^2 \\
=& \sum_{j=1}^d |u^j(x)|^2 - |x|^{-2} \sum_{j=1}^d \sum_{\ell=1}^d x_j x_\ell \left( u^\ell(x)\overline{u^j(x)}+\overline{u^\ell(x)}u^j(x) \right) \nonumber \\
& \hspace{5cm} + |x|^{-2} \sum_{\ell=1}^d \sum_{m=1}^d x_\ell x_m \overline{u^\ell(x)}u^m(x) \nonumber \\
=& \sum_{j=1}^d |u^j(x)|^2 - |x|^{-2} \sum_{\ell=1}^d \sum_{m=1}^d x_\ell x_m \overline{u^\ell(x)}u^m(x), \quad x\neq0 . \nonumber
\end{align}

Combining \eqref{differentiation across trajectory} with \eqref{radiation pf commutator}, we obtain
%
\begin{align*}
\left( u, i[H,\mathbb{A}]u\right) 
 =& 2\sum_{j=1}^d \left\| \chi_{\{x\neq0\}}|x|^{-1/2} \nabla_{j}^\perp u \right\|^2 \\
  &+ ((P\chi)(D_x)u, R_2 (P\chi)(D_x)u) + (u, i[V,\mathbb{A}] u) .
\end{align*}
We see that $\jap{x}^{1+\rho} [V,\mathbb{A}] \in \mathcal{B}(\mathcal{H})$ by \eqref{long range}, \eqref{short range} and Lemma \ref{composition lemma}.
According to \cite[Proposition 0.5.11]{Y}, the above formula and local $H$-smoothness of $\la x \ra^{-s}$ for $s>\frac{1}{2}$ imply that of \eqref{radiation op}.
\end{proof}

\section{Classical mechanics}\label{classical section}

In this section, we construct phase functions used for the definition of time-independent modifiers $J_\pm$ in \eqref{modified WOs}.
For the precise definition of $J_\pm$, see \eqref{time-independent modifiers}.

Let $\l_k(\x):\mathcal{U}_k\to\R$, $k=1,\dots,K$, be the functions in \eqref{G_k representation}.
The next proposition concerns the classical scattering problem with respect to the Hamiltonian $\l_k(\x)+\tilde V_\ell(x)$ on $T^*\mathcal{U}_k=\R_x^d\times\mathcal{U}_k$, where $\tilde V_\ell$ is a smooth extension of $V_\ell$ onto $\R^d$ such that $|\partial_x^\a \tilde V_\ell(x)|\leq C_\a' \la x \ra^{-\rho-|\a|}$ holds.
See \cite[Lemma 2.1]{N} for a concrete construction of $\tilde V_\ell$.

The proof of the following proposition is given by \cite[Section 2]{N3} (see also \cite{T1} and \cite{IsoKita}).

\begin{prop}\label{phase functions}
Let $\l_k(\x):\mathcal{U}_k\to\R$, $k=1,\dots,K$, be fixed. 
Then for any open set $U \Subset \mathcal{U}_k$ and $\e\in(0,2)$, there exist $R>0$ and smooth functions $\f_\pm^k(x,\x)$ defined on a neighborhood of
\begin{align*}
D_{k,\pm} = \{(x,\x)\in \R^d \times U \mid |x| \geq R, \ \pm \cos(x,\nabla\lambda_k(\x)) \geq-1+\e \},
\end{align*}
where
\begin{align*}
\cos(x,\nabla\l_k(\x)) := \frac{x\cdot\nabla\l_k(\x)}{|x| |\nabla\l_k(\x)|} ,
\end{align*}
such that
\begin{align}\label{eikonal}
\l_k(\nabla_x\f_\pm^k(x,\x)) + \tilde V_\ell(x) = \l_k(\x), \quad (x,\x) \in D_{k,\pm}.
\end{align}
Furthermore, $\f_\pm^k$ satisfy for $(x,\xi)\in D_{k,\pm}$
\begin{align} 
\left|\partial_x^\alpha\partial_\xi^\beta\left[\varphi_\pm^k(x,\xi)-x\cdot\xi\right]\right|
&\leq C_{\alpha\beta} \langle x\rangle^{1-\rho-|\alpha|}, \label{phase estimate} \\
\left|{}^t\nabla_x\nabla_\xi\varphi_\pm^k(x,\xi)-I \right|&<\frac{1}{2}. \label{phase estimate2}
\end{align}


\end{prop}

\section{Construction of Isozaki-Kitada modifiers}\label{IK section}

Let $\G\Subset\G(H_0)$ be fixed.
Let $\l_k \in C^\infty(\mathcal{U}_k)$, $k=1,\dots,K$, be as in \eqref{G_k representation}, 
and let $\f_{\pm}^{k}$ be  the phase functions constructed in Proposition \ref{phase functions} with setting $\e=\frac14$ and $U$ so that $\l_k^{-1}(\G)\Subset U \Subset \mathcal{U}_k$.

We take functions $\chi_k\in C_c^\infty(U;[0,1])$, $\eta\in C^\infty(\R^d)$ and $\s_\pm\in C^\infty(\R;[0,1])$ such that
\begin{align}
&\chi_k(\x)=1, \quad \x \in \lambda_k^{-1}(\G) , \label{chi 1} \\
&\eta(x)=
	\begin{cases}
		1 & \quad \text{if } |x|\geq 2R , \\
		0 & \quad \text{if } |x|\leq R ,
	\end{cases}
\label{eta 01} \\
&\s_\pm(\theta)=
	\begin{cases}
		1 & \quad \text{if } \pm\theta \geq \frac{1}{2}, \\ 
		0 & \quad \text{if } \pm\theta \leq -\frac{1}{2}, 
	\end{cases}
\label{sigma 01}\\
&\s_+(\theta)^2+\s_-(\theta)^2=1, \quad \theta\in\R, \label{sigma sum1}
\end{align}
where $R>0$ is the constant in Proposition \ref{phase functions}.
Then we define the Isozaki-Kitada modifiers $J_\pm^k$ associated with the pair $(P_k,\l_k,\mathcal{U}_k)$ by
\begin{align}\label{IK1}
J_\pm^k u(x) := 
(2\pi)^{-\frac{d}{2}} \int_{\T^d} e^{i\f_{\pm}^k(x,\x)} s_\pm^k(x,\x) \mathcal{F}u(\x) d\x ,
\end{align}
where
\begin{align*}
s_\pm^k(x,\x):=\eta(x)\s_\pm\left( \cos(x,\nabla\l_k(\x)) \right) P_k(\x) \chi_k(\x).
\end{align*}
We recall that $P_k(\x)$ is the projection matrix onto $\ker(H_0(\x)-\l_k(\x))$,
and note that $\supp s_\pm^k \subset D_{k,\pm}$ holds.
Their formal adjoints are given by
\begin{align*}
(J_\pm^k)^* u(x) = \mathcal{F}^* \left( (2\pi)^{-\frac{d}{2}} \sum_{y\in\Z^d} e^{-i\f_{\pm}^k(y,\cdot)} s_\pm^k(y,\cdot) u(y) \right) .
\end{align*}
Direct computations imply
\begin{align}\label{s is in S^0}
\sup_{(x,\x)\in\R^d\times\T^d}\jap{x}^{|\a|}|\partial_x^\a\partial_{\x}^\b s_\pm^k(x,\x)| < \infty ,
\end{align}
in particular \eqref{S^0} holds.

The next lemma follows from an analogue of the argument of calculus of Fourier integral operators (see \cite{IsoKita} and \cite{T1}).

\begin{lem}\label{J lem}
Let $k=1,\dots,K$ be fixed, and let $\rho>0$ be the constant in Assumption \ref{ass} (2). Then:

(1) $J_\pm^k$ are bounded operators on $\mathcal{H}$. 

(2) The operators
\begin{align}
\la x \ra^\rho \left( J_\pm^k (J_\pm^k)^* - s_\pm^k(x,D_x) s_\pm^k(x,D_x)^* \right), \label{J2}\\
\la x \ra^\rho \left( (J_\pm^k)^* J_\pm^k - s_\pm^k(x,D_x)^* s_\pm^k(x,D_x) \right) \label{J3}
\end{align}
are bounded on $\mathcal{H}$.

(3) For any $q\geq0$,
\begin{align}
\la x \ra^{-q} J_\pm^k \la x \ra^q, \label{JX} 
\end{align}
is bounded on $\mathcal{H}$.

(4) Suppose that $\psi = \psi(\x) \in C^\infty(\T^d;M_n(\C))$ commutes with $s_\pm^k(x,\x)$ for any $(x,\x)\in\Z^d\times\T^d$. 
Then
\begin{align}
\la x \ra^\rho [J_\pm^k, \psi(D_x)] \label{JXI}
\end{align}
is bounded on $\mathcal{H}$.
In particular, $[J_\pm^k, \psi(D_x)]$ are compact.

(5) If $k \neq \ell$, then $J_\pm^k (J_\pm^\ell)^*=0$, and $(J_\pm^k)^* J_\pm^\ell$ are compact on $\mathcal{H}$.
\end{lem}

\begin{proof}
(1) We compute 
\begin{align*}
J_\pm^k (J_\pm^k)^* u(x) = (2\pi)^{-d} \int_{\T^d} \sum_{y\in\Z^d} e^{i(\f_{\pm}^k(x,\x)-\f_{\pm}^k(y,\x))} s_\pm^k(x,\x) s_\pm^k(y,\x) u(y) d\x .
\end{align*}
We set $\f_{\pm}^k(x,\x)-\f_{\pm}^k(y,\x) = (x-y) \cdot \z(\x;x,y)$, where
\begin{align*}
\z(\x;x,y) := \int_0^1 \nabla_x \f_{\pm}^k(y+\theta(x-y),\x) d\theta.
\end{align*}
Then Proposition \ref{phase functions} implies that the mapping $\x \mapsto \z(\x;x,y)$ is a diffeomorphism from $U$ into $\z(U)$ for any $x,y \in \Z^d$.
%
%
Thus we have
\begin{align*}
J_\pm^k (J_\pm^k)^* u(x) 
= (2\pi)^{-d} \int_{\T^d} \sum_{y\in\Z^d} e^{i(x-y)\cdot\z} t_\pm^k(x,y,\z) u(y) d\z ,
\end{align*}
where
\begin{align*}
t_\pm^k(x,y,\z):=s_\pm^k(x,\x(\z;x,y)) s_\pm^k(y,\x(\z;x,y)) \left|\det\left( \frac{d\x}{d\z} \right)\right| .
\end{align*}
Since $\left| \frac{d\z}{d\x}(\x) - I \right| < \frac{1}{2}$ by Proposition \ref{phase functions}, \eqref{s is in S^0} implies $|\partial_{\z}^\a t_\pm^k(x,y,\z)|\leq C_\a$ for any $\a$. 
Therefore $J_\pm^k$ are bounded by Lemma \ref{PDO lemma}.

(2) 
The same argument as in (1) implies 
\begin{align*}
&\left( J_\pm^k (J_\pm^k)^* - s_\pm^k(x,D_x) s_\pm^k(x,D_x)^* \right) u(x) \\
=& (2\pi)^{-d} \int_{\T^d} \sum_{y\in\Z^d} e^{i(x-y)\cdot\z} r(x,y,\z) u(y) d\z ,
\end{align*}
where
\begin{align*}
r(x,y,\z)=t_\pm^k(x,y,\z) - s_\pm^k(x,\z) s_\pm^k(y,\z) .
\end{align*}
Since $|\partial_\x^\a r(x,\z,y)| \leq C_\a' \la x \ra^{-\rho}$, Lemma \ref{PDO lemma} implies the boundedness of \eqref{J2}.

The other case \eqref{J3} can be treated similarly if we consider the justification of PDO calculus; the argument using Poisson's summation formula as in \cite[Lemma 7.1]{N3} (see also \cite[Lemma 2.3]{T1}) implies
\begin{align*}
&\mathcal{F} (J_\pm^k)^* J_\pm^k \mathcal{F}^* f(\x) \\
 =& (2\pi)^{-d} \int_{\R^d} \int_{\T^d} e^{i(-\f_{\pm}^k(x,\x)+\f_{\pm}^k(x,\eta))} s_\pm^k(x,\x) s_\pm^k(x,\eta) f(\eta) d\eta dx
 + K_1 f(\x), \\
&\mathcal{F} s_\pm^k(x,D_x)^* s_\pm^k(x,D_x) \mathcal{F}^* f(\x) \\
 =& (2\pi)^{-d} \int_{\R^d} \int_{\T^d} e^{i x\cdot(-\x+\eta)} s_\pm^k(x,\x) s_\pm^k(x,\eta) f(\eta) d\eta dx
 + K_2 f(\x), 
\end{align*}
where $K_j$, $j=1,2$, is a smoothing operator in the sense that $\jap{D_x}^N K_j \in \mathcal{B}(\mathcal{H})$ for any $N>0$.
Then by changing variables $x\mapsto \int_0^1 \nabla_\x\f_\pm^k(x,\x+\theta(\eta-\x))d\theta$, PDO calculus on $\T^d$ implies the boundedness of \eqref{J3}.

(3) By a complex interpolation argument, it suffices to show (\ref{JX}) for $q\in 2\Z_+$.
Note that for $\a\in\Z_+^d$
\begin{align*}
& J_\pm^k x^\a u(x) \\
 =& (2\pi)^{-\frac{d}{2}} \int_{\T^d} e^{i\f_{\pm}^k(x,\x)}s_\pm^k(x,\x) i^{|\a|}\partial_{\x}^\a \mathcal{F}u(\x) d\x \nonumber\\
 =& (-i)^{|\a|} (2\pi)^{-\frac{d}{2}} \int_{\T^d} \partial_{\x}^\a (e^{i\f_{\pm}^k(x,\x)} s_\pm^k(x,\x) ) \mathcal{F}u(\x) d\x . \nonumber
\end{align*}
Then we learn for any $N\in\Z_+$,
\begin{align*}
J_\pm^k \jap{x}^{2N} u(x)
 = (2\pi)^{-\frac{d}{2}} \int_{\T^d} e^{i\f_{\pm}^k(x,\x)} (L^{N}s_\pm^k)(x,\x) \mathcal{F}u(\x) d\x ,
\end{align*}
where $L := \jap{\nabla_\x\f_\pm^k}^2 - i\Delta_\x\f_\pm^k - 2i\jap{\nabla_\x\f_\pm^k,\nabla_\x} - \Delta_\x$.
Since
\begin{align*}
\left|\partial_\x^\b(L^{N}s_\pm^k)(x,\x)\right| \leq C_{p,\b,N} \jap{x}^N
\end{align*}
for any $\b\in\Z_+^d$, we have the boundedness of \eqref{JX}.


(4) 
It suffices to show the boundedness of $\jap{D_\x}^\rho [\hat{J}_\pm^k, \psi(\x)]$ as an operator on $L^2(\T^d;\C^n)$, where $\hat{J}_\pm^k:=\mathcal{F}J_\pm^k\mathcal{F}^*$. 
Direct computation imply
\begin{align*}
&\la D_\x \ra^\rho [\hat{J}_\pm^k, \psi(\x)]f(\x) \\
=& (2\pi)^{-d} \sum_{x\in\Z^d}\int_{\T^d} e^{i(-x\cdot\x+\f_\pm^k(x,\y))} \jap{x}^\rho (\psi(\y)-\psi(\x)) s_\pm^k(x,\y) f(\y) d\y \nonumber \\
=& (2\pi)^{-d} \sum_{x\in\Z^d}\int_{\T^d} e^{i(-x\cdot\x+\f_\pm^k(x,\y))} \jap{x}^\rho \Psi_1(x,\y) s_\pm^k(x,\y) f(\y) d\y \nonumber \\
&+ (2\pi)^{-d} \sum_{x\in\Z^d}\int_{\T^d} e^{i(-x\cdot\x+\f_\pm^k(x,\y))} \jap{x}^\rho \Psi_2(x,\x,\y) s_\pm^k(x,\y) f(\y) d\y , \nonumber
\end{align*}
where
\begin{align*}
&\Psi_1(x,\y):= \psi(\y)-\psi(\nabla_x\f_\pm^k(x,\y)), \\
&\Psi_2(x,\x,\y):= \psi(\nabla_x\f_\pm^k(x,\y))-\psi(\x).
\end{align*}
The first term is treated similarly to (2), since $|\partial_\y^\a \Psi_1(x,\y)| \leq C_\a \jap{x}^{-\rho}$ by \eqref{phase estimate}. 
For the second term, we first employ the argument in the proof of boundedness of \eqref{J3} to replace the summation over $\Z^d$ by the integral on $\R^d$ modulo smoothing operators.
Then, since
\begin{align*}
\Psi_2(x,\x,\y) 
= (\nabla_x\f_\pm^k(x,\y)-\x)\cdot\int_0^1 \nabla_\x\psi(\x+\theta(\nabla_x\f_\pm^k(x,\y)-\x)) d\theta,
\end{align*}
we have
\begin{align*}
&(2\pi)^{-d} \int_{\R^d}\int_{\T^d} e^{i(-x\cdot\x+\f_\pm^k(x,\y))} \jap{x}^\rho \Psi_2(x,\x,\y) s_\pm^k(x,\y) f(\eta) d\eta dx \\
=& i(2\pi)^{-d} \int_{\R^d}\int_{\T^d} e^{i(-x\cdot\x+\f_\pm^k(x,\y))} a(\x,\y,x) f(\y) d\y dx ,
\end{align*}
where
\begin{align*}
a(\x,\y,x)=\nabla_x \cdot \left( \jap{x}^\rho s_\pm^k(x,\y) \int_0^1 \nabla_\x\psi(\x+\theta(\nabla_x\f_\pm^k(x,\y)-\x)) d\theta \right)
\end{align*}
satisfies $|\partial_\x^\a\partial_\y^\b\partial_x^\g a(\x,\y,x)|\leq C_{\a,\b,\g}$.
Finally we apply \cite[Theorem 2.1]{As-F} to obtain the boundedness of the second term.

(5) The first assertion follows from $s_k^\pm(x,\x)s_\ell^\pm(y,\x)=0$ for any $x$, $y$ and $\x$.
For the second assertion, we set $\psi_k\in C^\infty(\T^d;M_n(\C))$ so that $\psi_k(\x)=P_k(\x)$ on $\supp \chi_k$.
Then we use the equality $J_\pm^k=J_\pm^k \psi_k(D_x)$ and compactness of $[J_\pm^k,\psi_k(D_x)]$, which follows from (4).

\end{proof}

Now we prove the existence of the following (inverse) local wave operators
\begin{align}
 W^\pm(\mathcal{J}) :=& \slim e^{itH} \mathcal{J} e^{-itH_0} E_{H_0}(\G), \label{local mWO}\\
 I^\pm(\mathcal{J}) :=& \slim e^{itH_0} \mathcal{J}^* e^{-itH} E_H^{\textrm{ac}}(\G), \label{i local mWO}
\end{align}
for $\mathcal{J}=J_\#^k$ with $k=1,\dots,K$ and $\#\in\{+,-\}$.
Note that, if $\mathcal{J}$ is compact, then $W^\pm(\mathcal{J})=I^\pm(\mathcal{J})=0$.

We set $\tilde\chi_k \in C_c^\infty(\mathcal{U}_k)$ so that $\tilde\chi_k=1$ on $\supp\chi_k$.
Since
\begin{align*}
(P_k \tilde\chi_k)(D_x) J_\#^k - J_\#^k 
 = [(P_k \tilde\chi_k)(D_x), J_\#^k]
\end{align*}
is compact by Lemma \ref{J lem} (4), we have
\begin{align*}
W^\pm(J_\#^k) =& W^\pm((P_k \tilde\chi_k)(D_x) J_\#^k) , \\
I^\pm(J_\#^k) =& I^\pm((P_k \tilde\chi_k)(D_x) J_\#^k) ,
\end{align*}
and thus it suffices to show the existence of \eqref{local mWO} and \eqref{i local mWO} for
\begin{align*}
\mathcal{J}=(P_k \tilde\chi_k)(D_x) J_\#^k.
\end{align*}




\begin{lem}\label{HJ-JH_0 lem}
\begin{align*}
&(H (P_k \tilde\chi_k)(D_x)J_\pm^k - (P_k \tilde\chi_k)(D_x)J_\pm^k H_0)u(x)  \\
=& (2\pi)^{-d}\int_{\mathbb{T}^d}\sum_{y\in\mathbb{Z}^d}e^{i(\f_\pm^k(x,\x)-y\cdot\xi)}a_\pm^k(x,\xi)u(y)d\xi, \nonumber
\end{align*}
where
\begin{align}\label{s_a def}
&a_\pm^k(x,\x) \\
=& - i \y(x)\s_\pm'(\cos(x,\nabla_\x \l_k(\x))) \frac{|\nabla_\x \l_k(\x)|^2 - |x|^{-2}  (x\cdot\nabla_\x \l_k(\x))^2}{|x| |\nabla_\x \l_k(\x)|} P_k(\x) \chi_k(\x) \nonumber \\
& + r_\pm^k(x,\x) \nonumber
\end{align}
and $|\partial_\x^\b r_\pm^k(x,\x)| \leq C_\b \la x \ra^{-\min(1+\rho,2)}$.
\end{lem}

\begin{proof}
\textbf{Step 1.}
Let
\begin{align*}
g(x):=&(2\pi)^{-d}\int_{\T^d}e^{ix\cdot\x}H_0(\x)P_k(\x)\tilde\chi_k(\x)d\x \\
=&(2\pi)^{-d}\int_{\T^d}e^{ix\cdot\x}\lambda_k(\x)P_k(\x)\tilde\chi_k(\x)d\x .
\end{align*}
Then we learn
\begin{align*}
H_0 (P_k \tilde\chi_k)(D_x)J_\pm^k u(x)
=& (2\pi)^{-d}\int_{\mathbb{T}^d}\sum_{y\in\mathbb{Z}^d}e^{i(\f_\pm^k(x,\x)-y\cdot\xi)}a_\pm^{k,1}(x,\xi)u(y)d\xi,
\end{align*}
where
\begin{align*}
a_\pm^{k,1}(x,\xi) =& \sum_{y\in\Z^d} g(y) e^{i(\f_\pm^k(x-y,\x)-\f_\pm^k(x,\x))}s_\pm^k(x-y,\x) \\
=& \sum_{y\in\Z^d} g(y) e^{-iy\cdot\nabla_x\f_\pm^k(x,\x)} (1+R(x,y,\x)) s_\pm^k(x-y,\x),
\end{align*}
and
\begin{align*}
R(x,y,\x) := \exp \left[i\left( \f_\pm^k(x-y,\x)-\f_\pm^k(x,\x) + y\cdot\nabla_x\f_\pm^k(x,\x) \right)\right] -1.
\end{align*}
Since
\begin{align*}
&\left|\partial_\x^\b \left[\f_\pm^k(x-y,\x)-\f_\pm^k(x,\x) + y\cdot \nabla_x\f_\pm^k(x,\x) \right]\right| \\
=& \left| y \cdot \int_0^1 \partial_\x^\b \left(\nabla_x\f_\pm^k(x,\x) - \nabla_x\f_\pm^k(x-\theta y, \x) \right) d\theta \right| \\
=& \left| y \cdot \int_0^1 \left( \int_0^1 \partial_\x^\b \nabla_x^2\f_\pm^k(x-\phi\theta y, \x) d\phi \right) \theta y d\theta \right| \\
\leq& C_\b \la x \ra^{-1-\rho} \jap{y}^{3+\rho} ,
\end{align*}
we learn $\left|\partial_\x^\b R(x,y,\x)\right|\leq C_\b' \la x \ra^{-1-\rho} \jap{y}^{(3+\rho)\max\{1,|\b|\}}$, and thus
\begin{align*}
\left|\partial_\x^\b\sum_{y\in\Z^d} g(y) e^{-iy\cdot\nabla_x\f_\pm^k(x,\x)} R(x,y,\x) s_\pm^k(x-y,\x)\right|\leq C''_\b \la x \ra^{-1-\rho}.
\end{align*}
Furthermore, since \eqref{s is in S^0} implies the similar inequality
\begin{align*}
& \left|\partial_\x^\b \left[s_\pm^k(x-y,\x) - s_\pm^k(x,\x) + y \cdot \nabla_x s_\pm^k(x,\x)\right]\right| \\
=& \left| y \cdot\int_0^1 \left(\int_0^1 \partial_\x^\b\nabla_x^2 s_\pm^k(x-\phi\theta y,\x) d\phi\right) \theta y d\theta \right| \\
\leq& C_\b \jap{x}^{-2}\jap{y}^4 ,
\end{align*}
we have
\begin{align*}
&\sum_{y\in\Z^d} g(y) e^{-iy\cdot\nabla_x\f_\pm^k(x,\x)}s_\pm^k(x-y,\x) \\
=& \sum_{y\in\Z^d} g(y) e^{-iy\cdot\nabla_x\f_\pm^k(x,\x)} \left( s_\pm^k(x,\x) - y\cdot\nabla_x s_\pm^k(x,\x) \right) + O(\la x \ra^{-2}) \\
=& (\lambda_kP_k\tilde\chi_k)(\nabla_x\f_\pm^k(x,\x)) s_\pm^k(x,\x) - i \nabla_\x (\lambda_kP_k\tilde\chi_k)(\nabla_x\f_\pm^k(x,\x))\cdot\nabla_x s_\pm^k(x,\x) \\
 & \quad\quad + O(\la x \ra^{-2}).
\end{align*}
Thus we obtain
\begin{align*}
&a_\pm^{k,1}(x,\xi) \\
=& (\lambda_kP_k\tilde\chi_k)(\nabla_x\f_\pm^k(x,\x)) s_\pm^k(x,\x) - i \nabla_\x (\lambda_kP_k\tilde\chi_k)(\nabla_x\f_\pm^k(x,\x))\cdot\nabla_x s_\pm^k(x,\x) \\
 & \quad\quad + O(\la x \ra^{-\min(1+\rho,2)}).
\end{align*}

Similar computations imply that
\begin{align*}
V (P_k \tilde\chi_k)(D_x)J_\pm^k u(x)
=& (2\pi)^{-d}\int_{\mathbb{T}^d}\sum_{y\in\mathbb{Z}^d}e^{i(\f_\pm^k(x,\x)-y\cdot\xi)}a_\pm^{k,2}(x,\xi)u(y)d\xi, \\
(P_k \tilde\chi_k)(D_x)J_\pm^k H_0 u(x) 
=& (2\pi)^{-d}\int_{\mathbb{T}^d}\sum_{y\in\mathbb{Z}^d}e^{i(\f_\pm^k(x,\x)-y\cdot\xi)}a_\pm^{k,3}(x,\xi)u(y)d\xi,
\end{align*}
where
\begin{align*}
&a_\pm^{k,2}(x,\xi) \\
=& V(x) \left((P_k\tilde\chi_k)(\nabla_x\f_\pm^k(x,\x)) s_\pm^k(x,\x) - i \nabla_\x (P_k\tilde\chi_k)(\nabla_x\f_\pm^k(x,\x))\cdot\nabla_x s_\pm^k(x,\x) \right)  \\
 & \quad\quad + O(\la x \ra^{-\rho-\min(1+\rho,2)}) , \\
&a_\pm^{k,3}(x,\xi) \\
=& \lambda_k(\x) (P_k\tilde\chi_k)(\nabla_x\f_\pm^k(x,\x)) s_\pm^k(x,\x) - i \lambda_k(\x) \nabla_\x (P_k\tilde\chi_k)(\nabla_x\f_\pm^k(x,\x))\cdot\nabla_x s_\pm^k(x,\x) \\
 & \quad\quad + O(\la x \ra^{-\min(1+\rho,2)}).
\end{align*}

\textbf{Step 2.}
Step 1 implies
\begin{align*}
&a_\pm^k(x,\x) \\
=& (P_k\tilde\chi_k)(\nabla_x\f_\pm^k(x,\x)) s_\pm^k(x,\x) (\lambda_k(\nabla_x\f_\pm^k(x,\x))+V(x)-\lambda_k(\x)) \\
&-i \nabla_\x (P_k\tilde\chi_k)(\nabla_x\f_\pm^k(x,\x))\cdot \nabla_x s_\pm^k(x,\x) (\lambda_k(\nabla_x\f_\pm^k(x,\x))+V(x)-\lambda_k(\x)) \\
&-i (P_k\tilde\chi_k)(\nabla_x\f_\pm^k(x,\x))\nabla_\x \lambda_k(\nabla_x\f_\pm^k(x,\x))\cdot\nabla_x s_\pm^k(x,\x) \\
&  +O(\la x \ra^{-\min(1+\rho,2)}).
\end{align*}
The first and second terms are of order $\jap{x}^{-1-\rho}$ by \eqref{eikonal} and \eqref{short range}.
Moreover simple computations imply that, setting $v:=\nabla_\x\lambda_k(\x)$,
\begin{align*}
&\nabla_x s_\pm^k(x,\x) \\
=& \eta(x)\s_\pm'\left( \cos(x,v) \right) \left( \frac{1}{|x||v|}v - \frac{x\cdot v}{|x|^3|v|}x \right) P_k(\x) \chi_k(\x) +O(\jap{x}^{-\infty}) ,
\end{align*}
and therefore
\begin{align*}
&a_\pm^k(x,\x) \\
=& -i (P_k\tilde\chi_k)(\x)\nabla_\x \lambda_k(\x)\cdot\nabla_x s_\pm^k(x,\x) + O(\la x \ra^{-\min(1+\rho,2)}) \\
=& -i \eta(x)\s_\pm'\left( \cos(x,v) \right) \left(\frac{|v|}{|x|}-\frac{(x\cdot v)^2}{|x|^3|v|}\right) P_k(\x) \chi_k(\x) + O(\la x \ra^{-\min(1+\rho,2)}) .
\end{align*}
Here we have used \eqref{phase estimate} in the first equality to replace $\nabla_x\f_\pm^k(x,\x)$ by $\x$.

\end{proof}

\begin{prop}\label{existence of local WO}
For any $k=1,\dots,K$, there exist the limits \eqref{local mWO} and \eqref{i local mWO} with $\mathcal{J}=J_\pm^k$.
\end{prop}

\begin{proof}
We only prove the existence of \eqref{local mWO}, since the other is done in the same way.

We may assume $\rho<1$ without loss of generality.
The standard argument of existence of (modified) wave operators (see, e.g., \cite[Lemmas 10.2.1 and 10.2.2, Theorem 0.5.4]{Y} and \cite[Theorem XIII. 24]{R-S4}) implies that it suffices to prove that $H (P_k \tilde\chi_k)(D_x)J_\pm^k - (P_k \tilde\chi_k)(D_x)J_\pm^k H_0$ is a finite sum of  the form $G_j^* B_j G'_j$ with $G_j$ (resp.\ $G'_j$) being $H$-(resp.\ $H_0$-) smooth in $\Gamma$ and $B_j\in\mathcal{B}(\mathcal{H})$.

We set
\begin{align*}
a_j^k(x,\x) =& \y(x)|x|^{-\frac{1}{2}}\left(\partial_{\x_j}\l_k(\x)-|x|^{-2} x_j ( x\cdot\nabla_\x\l_k(\x))\right) P_k(\x) \tilde\chi_k(\x), \\
b_{\pm}^k(x,\x) =& -i \y(x)\s_\pm'\left( \cos(x,\nabla_\x\lambda_k(\x)) \right) P_k(\x) \chi_k(\x).
\end{align*}
Then we observe that 
\begin{align*}
a_j^k(x,D_x)=\y(x)|x|^{-\frac12}\nabla_{k,j}^\perp ,
\end{align*}
where $\nabla_{k,j}^\perp$ is as in Proposition \ref{prop_rad_est}.
Moreover we have by the definition \eqref{s_a def} of $a_\pm^k(x,\x)$
\begin{align*}
a_\pm^k(x,\x)=b_\pm^k(x,\x) \sum_{j=1}^d a_j^k(x,\x)^2 + r_\pm^k(x,\x) ,
\end{align*}
where $\partial_\x^\a r_\pm^k(x,\x) = O(\jap{x}^{-2})$.

We take functions $\tilde{\tilde\chi}_k \in C_c^\infty(\mathcal{U}_k)$ and $\tilde\s_\pm(\theta) \in C^\infty(\R)$ such that
\begin{align*}
&\tilde\s_\pm(\theta)=
	\begin{cases}
		1 & \quad \text{if } \pm\theta \geq -\frac{1}{2}, \\ 
		0 & \quad \text{if } \pm\theta \geq -\frac{3}{4}, 
	\end{cases}
\\
&\tilde{\tilde\chi}_k(\x) = 1 , \quad \x\in \supp\tilde\chi_k .
\end{align*}
We set
\begin{align*}
\tilde s^k(x,\x) =& \eta(x) P_k(\x) \tilde{\tilde\chi}_k(\x), \\
\tilde \f_{\pm}^k(x,\x) =& \y(x)\tilde\s_\pm(\cos(x,\nabla\lambda_k(\x))) \f_{\pm}^k(x,\x) \\
& \quad + \left( 1 - \y(x)\tilde\s_\pm(\cos(x,\nabla\lambda_k(\x))) \right) x\cdot\x,
\end{align*}
and
\begin{align*}
\tilde J_\pm^k u(x) =& (2\pi)^{-\frac{d}{2}} \int_{\T^d} e^{i\tilde\f_{\pm}^k(x,\x)} \tilde s^k(x,\x) \mathcal{F}u(\x) d\x , \\
A_{\pm,j}^ku(x)=&(2\pi)^{-\frac{d}{2}}\int_{\T^d} e^{i\tilde\f_\pm^k(x,\x)} a_j^k(x,\x) \mathcal{F}u(\x)d\x, \\
C_{\pm,j}^ku(x)=&(2\pi)^{-\frac{d}{2}}\int_{\T^d} e^{i\f_\pm^k(x,\x)} b_\pm^k(x,\x)a_j^k(x,\x)^2 \mathcal{F}u(\x)d\x.
\end{align*}
Then it follows from the same argument as Lemma \ref{J lem} (2) that
\begin{align*}
&\tilde J_\pm^k (\tilde J_\pm^k)^* = \tilde s^k(x,D_x)^2 +R_{\pm,j,1}^k, \\
&(\tilde J_\pm^k)^*A_{\pm,j}^k = a_j^k(x,D_x)+R_{\pm,j,2}^k, \\
&(\tilde J_\pm^k)^*C_{\pm,j}^k = a_j^k(x,D_x) b_\pm^k(x,D_x) a_j^k(x,D_x)+R_{\pm,j,3}^k, 
\end{align*}
where $\jap{x}^{\frac{1+\rho}{2}} R_{\pm,j,\ell}^k \jap{x}^{\frac{1+\rho}{2}} \in \mathcal{B}(\HS)$, $\ell=1$, $2$, $3$.
Moreover we learn by the argument in Lemma \ref{J lem} (4) that
\begin{align*}
&\tilde s^k(x,D_x)^2 A_{\pm,j}^k = A_{\pm,j}^k + R_{\pm,j,4}^k, \\
&\tilde s^k(x,D_x)^2 C_{\pm,j}^k = C_{\pm,j}^k + R_{\pm,j,5}^k, 
\end{align*}
where $\jap{x}^{\frac{1+\rho}{2}} R_{\pm,j,\ell}^k \jap{x}^{\frac{1+\rho}{2}}\in \mathcal{B}(\HS)$, $\ell=4$, $5$.
Thus we have, modulo operators of the form $\jap{x}^{-\frac{1+\rho}{2}}B\jap{x}^{-\frac{1+\rho}{2}}$ with $B\in\mathcal{B}(\mathcal{H})$,
\begin{align*}
&H (P_k \tilde\chi_k)(D_x)J_\pm^k - (P_k \tilde\chi_k)(D_x)J_\pm^k H_0 \\
\equiv& \sum_{j=1}^d C_{\pm,j}^k \\
\equiv& \sum_{j=1}^d \tilde s^k(x,D_x)^2 C_{\pm,j}^k \\
\equiv& \sum_{j=1}^d \tilde J_\pm^k (\tilde J_\pm^k)^* C_{\pm,j}^k \\
\equiv& \sum_{j=1}^d \tilde J_\pm^k a_j^k(x,D_x) b_\pm^k(x,D_x) a_j^k(x,D_x) \\
\equiv& \sum_{j=1}^d \tilde J_\pm^k (\tilde J_\pm^k)^* A_{\pm,j}^k b_\pm^k(x,D_x) a_j^k(x,D_x) \\
\equiv& \sum_{j=1}^d \tilde s^k(x,D_x)^2 A_{\pm,j}^k b_\pm^k(x,D_x) a_j^k(x,D_x) \\
\equiv& \sum_{j=1}^d A_{\pm,j}^k b_\pm^k(x,D_x) a_j^k(x,D_x) .
\end{align*}

Since $b_\pm^k(x,D_x) \in \mathcal{B}(\HS)$ and Proposition \ref{prop_rad_est} implies $a_j^k(x,D_x)$ is $H_0$-smooth on $\G$,
it remains to prove that $A_{\pm,j}^k$ is $H$-smooth on $\G$.
However, the proof is completed if we observe that $a_j^k(x,D_x)$ and $\jap{x}^{\frac{1+\rho}{2}}$ are $H$-smooth on $\G$ and that
\begin{align*}
(A_{\pm,j}^k)^* A_{\pm,j}^k=a_j^k(x,D_x)^*a_j^k(x,D_x) + R''_j,
\end{align*}
where $\jap{x}^{\frac{1+\rho}{2}}R_j''\jap{x}^{\frac{1+\rho}{2}}\in\mathcal{B}(\mathcal{H})$.
\end{proof}


%
%

\section{Proof of Theorem \ref{main}} \label{proof section}


We set
\begin{align}\label{time-independent modifiers}
J_\pm := \sum_{k=1}^K J_\pm^k ,
\end{align}
where $J_\pm^k$'s are given by \eqref{IK1}.
Then Proposition \ref{existence of local WO} implies the existence of the modified wave operators \eqref{modified WOs}.
The proof of the intertwining property is skipped since it is easily proved.



\begin{prop} \label{nonstationary prop}
$W^\pm(J_\mp) = I^\pm(J_\mp) = 0$.
\end{prop}

\begin{proof}
For the first assertion, it suffices to prove $\lim_{t\to\pm\infty} J_\mp^k e^{-itH_0}u=0$ for any $u$ satisfying
\begin{align*}
(P_k\chi_k)(D_x)u=u.
\end{align*}
We easily see that
\begin{align*}
&J_\mp^k e^{-itH_0} u(x) \\
=&(2\pi)^{-\frac{d}{2}} \int_{\T^d} e^{i(\f_\mp^k(x,\x)-t\lambda_k(\x))} \eta(x)\s_\mp\left( \cos(x,\nabla\lambda_k(\x)) \right) \mathcal{F}u(\x)d\x.
\end{align*}
The estimate \eqref{phase estimate} and the conditions \eqref{eta 01} and \eqref{sigma 01} imply there is a constant $c>0$ such that on the support of the integrand
\begin{align*}
|\nabla_\x \f_\mp^k(x,\x)-t\nabla\lambda_k(\x)|
\geq& |x-t\nabla\lambda_k(\x)| - |x-\nabla_\x\f_\mp^k(x,\x)| \\
\geq& \sqrt{\frac{1-\cos(x,\pm\nabla\lambda_k(\x))}{2}}|x| |t\nabla\lambda_k(\x)| - C\jap{x}^{1-\rho} \\
\geq& c (|x| + |t| |\nabla\lambda_k(\x)|)
\end{align*}
for sufficiently large $\pm t \geq 0$.
The non-stationary phase method implies that
\begin{align*}
|J_\mp^k e^{-itH_0} u(x)|\leq C_{N}(1+|x|+|t|)^{-N} , \quad x\in\Z^d,\ \pm t \geq0,
\end{align*}
for any $N\geq1$.
Thus we obtain $\|W^\pm(J_\mp)u\| = 0$.

For the other assertion $I^\pm(J_\mp) = 0$, the intertwining property implies
\begin{align*}
I^\pm(\mathcal{J})=I^\pm(\mathcal{J}) E_H(\G) = E_{H_0}(\G) I^\pm(\mathcal{J}).
\end{align*}
Thus we learn that for any $v\in\mathcal{H}$
\begin{align*}
(I^\pm(J_\mp)u, v) =& ( E_{H_0}(\G) I^\pm(J_\mp)u, v) \\
=& \lim_{t\to\pm\infty} (e^{itH_0} J_\mp^* e^{-itH} E_{H}^{\textrm{ac}}(\G) u, E_{H_0}(\G) v) \\
=& \lim_{t\to\pm\infty} ( E_{H}^{\textrm{ac}}(\G) u, e^{itH} J_\mp e^{-itH_0} E_{H_0}(\G) v) \\
=& ( E_{H}^{\textrm{ac}}(\G) u, W^\pm(J_\mp) v) \\
=& 0
\end{align*}
by the first assertion.
\end{proof}

\begin{prop}
For any $u\in\mathcal{H}$,
\begin{align}
\|W^\pm(J_\pm)u\| =& \|E_{H_0}(\G)u\|, \label{p isom}\\
\|I^\pm(J_\pm)u\| =& \|E_{H}^{\textrm{ac}}(\G)u\| . \label{i p isom}
\end{align}
\end{prop}

\begin{proof}
We learn
\begin{align*}
\|W^\pm(\mathcal{J})u\|^2
= \lim_{t\to\pm\infty} \| \mathcal{J}e^{-itH_0}E_{H_0}(\G)u \|^2
= \lim_{t\to\pm\infty} \left( u_t, \mathcal{J}^*\mathcal{J} u_t \right) ,
\end{align*}
where $u_t:=e^{-itH_0}E_{H_0}(\G)u$.
Thus Lemmas \ref{adjoint lemma}, \ref{composition lemma}, \ref{J lem} (2), (5) and \eqref{Rep of E_H0}, \eqref{chi 1}, \eqref{sigma sum1} imply
\begin{align*}
&\|W^\pm(J_+)u\|^2 + \|W^\pm(J_-)u\|^2 \\
=& \lim_{t\to\pm\infty} \left( u_t, (J_+^*J_+ + J_-^*J_-) u_t \right) \\
=& \lim_{t\to\pm\infty} \left( u_t, \left(\sum_{k=1}^K (J_+^k)^* J_+^k + (J_-^k)^* J_-^k\right) u_t \right) \\
=& \lim_{t\to\pm\infty} \left( u_t, 
\left( \sum_{k=1}^K s_+^k(x,D_x) s_+^k(x,D_x)^* + s_-^k(x,D_x) s_-^k(x,D_x)^* \right) u_t \right) \\
=& \lim_{t\to\pm\infty} \left( u_t, 
\y(x)^2 \sum_{k=1}^K (P_k\chi_k^2)(D_x) u_t \right) \\
=& \lim_{t\to\pm\infty} \left( u_t, 
\y(x)^2 u_t \right) \\
=& \|E_{H_0}(\G)u \|^2 .
\end{align*}
Here we have used \eqref{Rep of E_H0} and \eqref{eta 01} to obtain $\sum_{k=1}^K (P_k\chi_k^2)(D_x) E_{H_0}(\G) = E_{H_0}(\G)$ and compactness of $1-\y(x)^2$.
Therefore we have the first equality \eqref{p isom} by Proposition \ref{nonstationary prop}.

The other equality \eqref{i p isom} is obtained by the similar argument and the compactness of $\psi(H)-\psi(H_0)$ for $\psi\in C_c^\infty(\R)$.

\end{proof}

It remains to prove the completeness of \eqref{modified WOs}.
However it is proved by the existence of $I^\pm(J_\pm)$ and \eqref{i p isom}.


\end{document}